\documentclass[12pt, reqno]{amsart}
\usepackage{geometry}
\usepackage{amsthm}
\usepackage{amssymb}
\newtheorem{definition}{Definition}[section]
\newtheorem{theorem}{Theorem}[section]

\usepackage{enumerate}
\usepackage{graphicx}
\graphicspath{{./images/}}
\usepackage[T1]{fontenc}
\usepackage[utf8]{inputenc}
\numberwithin{equation}{section}
\usepackage[nomarkers]{endfloat}
\usepackage{hyperref}

\title{ Another Asymptotic Notation : ``Almost'' }
\author{Nabarun Mondal}
\address{Arcesium India Private Limited, India , Hyderabad }
\email{mondal@arcesium.com}
\thanks{Nabarun Mondal : 
Dedicated to my late professor Dr. Prashanta Kumar Nandi. \\ 
Dedicated to my parents. \\
Big thanks to:- Abhishek Chanda : You always have been constant support. \\
In Memory of : Dhrubajyoti Ghosh. Dear Dhru, rest in peace.
}

\author{Partha P. Ghosh}
\address{Microsoft India, Hyderabad }
\email{parthag@microsoft.com}
\thanks{ Partha. P. Ghosh : 
Dedicated to my parents and family without their presence we are nothing. \\
}


\begin{document}


\begin{abstract}
Asymptotic notations are heavily used while analysing runtimes of algorithms.
Present paper argues that some of these usages are non trivial, therefore 
incurring  errors in communication of ideas. After careful reconsideration 
of the various existing notations a new notation is proposed. This notation has
similarities with the other heavily used notations like Big-Oh, Big Theta, 
while being more accurate when describing the order relationship.
It has been argued that this notation is more suitable for describing 
algorithm runtime than Big-Oh.  
\end{abstract}

\maketitle

\begin{section}{The Problem At Hand}\label{intro}  
  
Describing the exact runtime of an algorithm \cite{ia}\cite{ahu}\cite{dk} 
is not a trivial task. Hence, the effort is directed upon finding a 
function that approximates the runtime.   
Asymptotic notations also have been heavily used in the random graph theory
\cite{rss}\cite{ijpam}\cite{im}.
To describe the approximate runtime for the algorithms, 
the concept of asymptotic notations were borrowed from Number Theory 
\cite{pb}\cite{el}\cite{cm}.

In this section we discuss the most common of the asymptotic notations, 
introduced by Bachmann-Landau \cite{pb}\cite{el},
and describe the semantic problem with the 
overly used Big-Oh notation \cite{ia}\cite{ahu}\cite{dk}\cite{cm}.

\begin{definition}\label{b-o}
\textbf{Big-Oh : O(.) }

Let $f(n)$ and $g(n)$ be functions such that :-
$$
\exists k>0 \; \exists n_0 \; \forall n>n_0 \; |f(n)| \leq |g(n)\cdot k| 
$$
then $f(n) \in O( g(n) )$ or with some abuse of notation $f(n) = O(g(n))$.
\end{definition}   
Informally this stands for $f$ is bounded above by  $g$ (up to constant factor) asymptotically.

\begin{definition}\label{s-o}
\textbf{Small-Oh : o(.) }

Let $f(n)$ and $g(n)$ be functions such that :-
$$
\forall k>0 \; \exists n_0 \; \forall n>n_0 \; |f(n)| \le k\cdot |g(n)|
$$
then $f(n) \in o( g(n) )$ or with some abuse of notation $f(n) = o(g(n))$.
\end{definition}   
Informally this stands for $f$ is dominated by  $g$ (up to constant factor) asymptotically.

\begin{theorem}\label{O-o-relation}  
\textbf{Relation between Big-O and Small-o.}

If $f \in o(g)$ then $f \in O(g)$.
\end{theorem}
\begin{proof}
We compare the definitions \eqref{s-o} and \eqref{b-o}.
Note that the $f \in o(g) \implies \exists k $ for which  $|f(n)| \leq |g(n)\cdot k|$.
Hence the given is proved.
\end{proof}
  
This theorem \eqref{O-o-relation} creates a confusion about 
the way mathematicians use the big-oh notation $f = O(g)$
and computer engineers use them. 
Let $f \in O(g)$, then from Computer Engineering standpoint :- 
\begin{center}
\emph{ ``growth rate of f is of the order of g''} 
\end{center}
and it's precise mathematical meaning:-
\begin{center}
\emph{``growth rate of f is less than or equal to the order of g''}
\end{center}
We now demonstrate the confusion with a real example.
For any algorithm $\mathcal{A}$ let the average runtime complexity be a function 
`$A(n)$' and worst case runtime complexity be another function `$W(n)$' 
where `$n$' be the input size. 
In the  \emph{strictly mathematical} way $A \in O( W )$.
However, this becomes counter intuitive as 
the intuitive notion of big-$O$ as 
\emph{the order of the expression} as it is majorly used in Engineering. 

Example of Quick sort can be used to demonstrate the point. 
$A(n) \approx nLog_2n$ and $W(n) \approx n^2 $.
\emph{Strictly mathematically} :-
$$
n Log_2 (n) = O ( n^2 ) = O ( n^3 ) = O ( n^k ) \; ; \; k \ge 2  
$$
because 
$$
n Log_2 (n) = o ( n^2 ) = o ( n^3 ) = o ( n^k ) \; ; \; k \ge 2
$$
and $f = o(g) \Rightarrow f = O(g)$.
However, clearly the order-of $n Log_2 (n)$ is not the order of 
$n^2$ or $n^3$ etc, and the confusion arises because technically $O(g)$ family 
is not an open family of functions 
(it depicts the relation `$\le$' in some sense). 
The family of functions $o(g)$ is however an open family, 
depicting the relation `$<$' \cite{cm}\cite{dk}. 
Simply put Big-Oh is not the  \emph{of the order of}
function family.  
\end{section}

\begin{section}{Stricter Asymptotic Bounds}\label{}

This problem of confusion on notation with intuition 
can be avoided by using any stricter notion, 
which  should match with the  intuition of \emph{growth order}.
One such Bachmann-Landau notation is the Big-Theta notation 
which is described next.

\begin{definition}\label{b-t}
\textbf{Big-Theta : $\Theta$(.) }

Let $f(n)$ and $g(n)$ be functions such that :-
$$
\exists k_1>0 \; \exists k_2>0 \; \exists n_0 \; \forall n>n_0
$$
$$
g(n) \cdot k_1 \leq f(n) \leq g(n) \cdot k_2
$$
then $f(n) \in \Theta( g(n) )$ or with abuse of notation $f(n) = \Theta(g(n))$.
\end{definition}   
This stands for  $f$ is  bounded both above and below by  $g$ asymptotically.
For example given that the function $f(n) = n^3 + 1000n^2 + n + 300$ then,
$f \in o(n^5) $ as well as $f \in O(n^5) $.
But $f \in \Theta(n^3)$ so is $f \in \Theta(n^3 + n^2)$, so on and so forth.
Therefore, the notation $\Theta(.)$ allows one to be far more accurate than 
that of the more popular $O(.)$ notation, but still, not enough accurate.  
One simple example to show this would be $f(n) = 2 - sin(n)$ while $g(n) = c$
a constant function with $c > 0 $.
Clearly $f = \Theta(g)$ but, then, accuracy  is lost here.

One way to resolve this loss of accuracy is to 
find an equivalence or similar notation, 
the real \emph{on the order of} notation as defined next.

\begin{definition}\label{e}
\textbf{On the Order Of : $ \sim $ }

Let $f(n)$ and $g(n)$ be functions such that :-
$$
\forall \varepsilon>0\;\exists n_0\;\forall n>n_0\;\left|{f(n) \over g(n)}-1\right|<\varepsilon
$$
then, $f \sim g$.
\end{definition}
The problem with definition \ref{e} is that this is much of a strict bound. 
It might not be optimal to find a function $g$ that approximates $f$ 
till the exact same growth order.
For example take the function $f(n) = 3n^3 + 2n^2 + n $.
We have $g(n) = 3n^3 $ which makes $f \sim g$ but for another 
$g'(n) = n^3$ $f \not \sim g'$.
Constant factors should be ignored, in this case it is not.
We note that, however, another way of representing the same notation is:-
$$
f \sim g \implies \forall n > n_0 \; \frac{f(n)}{g(n)} \to 1
$$  
Replacing the constant $1$ with a varying parameter $K$ 
makes it into a better asymptotic notation independent of the constant factor.
But that factor, should never be 0 or $\infty$.

Therefore, we can define a new, but similar asymptotic notation \emph{``Almost''} .
\begin{definition}\label{a}
\textbf{Almost : a(K,.). }

A function $f(n)$ is almost $g(n)$ or $f \in a(K, g )$ iff:-
$$
\lim_{n \to \infty} \frac{f(n)}{g(n)} = K \; ; \;  0 < \; K < \infty 
$$ 
exists.
\end{definition}

\end{section}

\begin{section}{Properties of ``Almost'' }\label{prop}
In this section we establish properties of ``almost'' 
with theorems which relates the notation of
\emph{almost} (definition \ref{a}) with other notations.

The trivial properties first:-
\begin{enumerate}
\item { $f \in a(1,f)$ , that is $f \; a \; f$ . }
\item { $f \in a(K,g) \implies g \in a (1/K , f )$ 
which is $f \; a \; g \implies g \; f \; a$. }
\item { $f \in a(K_1,g)$ and  $g \in a (K_2 , h )$ then 
$f \in a (K_3 , h )$ with 
$K_3 = K_1 K_2$ that is 
$f \; a \; g $ and $g \; a \; h \implies f \; a \; h$. 
}
\item { $f \in a(1,g) \implies f \sim g$ . }
\end{enumerate}

\begin{theorem}\label{a-e}
\textbf{ $a(.,.)$ is Equivalence Class Relation. }

Functions related using the notation $a(.,.)$ 
(definition \ref{a}) are an equivalence class of functions. 
\end{theorem}
\begin{proof}
The first three trivial properties demonstrates this.
\end{proof}

\begin{theorem}\label{a-O}
\textbf{Relation of $a(.)$  with $O(.)$. }

If $f \in a(K, g)$ then $f \in O(g)$ but not vice versa.
\end{theorem}
\begin{proof}
If we have
$$
\lim_{n \to \infty} \frac{f(n)}{g(n)} = K
$$
We note that $f \in a(k,g)$ implies $ 0 < K $,
while $f \in O(g)$ implies $ \exists k \; ; \; 0 \le k $.
Which establishes the theorem.
\end{proof}

\begin{theorem}\label{at}
\textbf{Relation of $a(.)$  with $\Theta(.)$. }

\begin{enumerate}
\item { If $f \in a(K,g)$ then $f \in \Theta(g)$.}
\item { If $f \in \Theta(g)$ then it is not necessary to have $f \in a(K,g)$.}
\item { Therefore, relation $a(.)$ is a subset of $\Theta(.)$.}
\end{enumerate}

\end{theorem}
\begin{proof}
By definition we have for $n > n_0$ $f(n) = Kg(n)$.
That is, obviously $f <(K+1)g$ and clearly $f > (K-1)g$.
Which means, tallying with the definition \ref{b-t},
we have $k_1 = K - 1$ and $k_2 = K + 1$.  
Which proves the first part.

We show the second part by bringing in one example. 
Let $f(n) = 2 - sin(n)$. Clearly then $\text{sup}(f(n)) = 3$ and 
$\text{inf}(f(n)) = 1$. Obviously then we can define $g(n)=1$, 
with $k_1 = 1$ with $k_2 = 3$, so that :-
$$
k_1 g \le f \le k_2 g  
$$ 
and therefore $f \in \Theta (1)$. 
However, no limit exists for the ratio $f(n)/g(n)$ with $n \to \infty$, 
and therefore $f \not \in a(k,g)$.

These show that there are functions $f = \Theta(g)$
but $f \ne a(,g)$, hence, the relation between the family 
is depicted as $a \subset \Theta$.
\end{proof}

These theorems establish that the usage of \emph{Almost} is
more appropriate from a strict computational aspect.
In general both $f_1(n) = n^2 + n $ and $f_2(n) =  6n $ 
being $O(n^2)$ is a confusing aspect, with respect to computation.
In no sense these functions $f_1,f_2$ can be made computationally equivalent.    
While the notion of $\Theta()$ solves it in some sense.
Almost is the subset of the $\Theta()$ where a strict limit on 
problem size parameter $n \to \infty $ can be formally established.
That essentially tells that : ``\emph{at problem size infinity (asymptote) the functions 
will be costing equivalent amount of computation} ; thus, they are \emph{`asymptotically equivalent'}.
That is indeed the truest definition of asymptotic equivalence.
When that limit fails to exist, the fallback is still the  $\Theta()$,
but this time, the implication stays very clear: 

One is strictly \emph{asymptotically bounded},
by another, they are not \emph{asymptotically equivalent}. 
  
\end{section}

\begin{section}{ Summary : Advantage of using ``almost'' }\label{adv}
There is no way one can state $f \in a(K,g)$,
unless, they are comparable to the limit, as defined by definition \ref{a}.

Therefore, the biggest advantage of the proposed notation is less confusion 
in the usage of the notation, because \emph{almost} is an equivalence notation.
Take for example $n^2 \in O(n^3)$ but $n^3 \not \in O(n^2)$.
But clearly $n^2 \not \in a(K,n^3)$.
 
In the same note, one can not in general write $A = a(K,W)$ 
like in section \ref{intro} , unless, of course $A,W$ are really comparable,
as in merge-sort. In specificity, we can not use the notation $a(K,.)$ 
to compare $f = nLog_2(n)$ and $g = n^2$, in case of Quick Sort. 

Given the limit at infinity exists, \emph{almost} becomes both the tight 
upper and lower bound, that is $\Theta(.)$ from definition \ref{b-t} 
using the theorem \ref{at}.
On the other hand, if no such limit exists (example given is theorem \ref{at}), 
we can still use the original $\Theta(.)$ notation.

In summary, the notation $a(K,.)$ has advantages borrowed from all the notations,
without any of their shortcomings, as long as it can be defined. 
Also, this is not as much a tight bound compare to the 
notation $\sim$ (definition \ref{e}).

\end{section}

\end{document}